\documentclass[12pt,notitlepage]{article}
\usepackage{amssymb}
\usepackage{amsthm}
\usepackage{amsmath}
\usepackage{bm}

\usepackage{graphicx}
\usepackage{tikz}
\usetikzlibrary{decorations.markings}
\usepackage{subcaption}
\usepackage{algorithm}
\usepackage{stmaryrd}
\usepackage[onehalfspacing]{setspace}
\usepackage[includehead]{geometry}
\usepackage{booktabs} 
\usepackage{color}
\usepackage{natbib}
\usetikzlibrary{decorations.text}
\usepackage{hyperref}
\hypersetup{
    colorlinks=true,
    citecolor=blue,
    linkcolor=blue,
    filecolor=magenta,      
    urlcolor=cyan,
    pdftitle={Overleaf Example},
    }
\usepackage{enumerate}
\usetikzlibrary{decorations.markings, arrows.meta}
\allowdisplaybreaks[4]
\usepackage{subcaption}
\usepackage{dutchcal}
\usepackage{tablefootnote}

\usepackage{algorithm}
\usepackage{algpseudocode}

\outer
\def\proclaim #1: #2\par{\medbreak
  \noindent{#1:\enspace}{\sl#2}\par
  \ifdim\lastskip<\bigskipamount \removelastskip\penalty55\medbreak\fi}
  \usepackage{amsthm}
\newtheorem{theorem}{Theorem}

\newtheorem*{theorem*}{Theorem}

\newtheorem{definition}{Definition}
\newtheorem{example}{Example}

\newtheorem{proposition}{Proposition}
\newtheorem{remark}{Remark}

\usepackage{indentfirst}

\newcounter{my}

\newcounter{my2}

\newcounter{my3}

\newcounter{my4}

\newcounter{my5}

\newcounter{my6}

\setcitestyle{authoryear}

\title{Mechanism Design in Max-Flows}

\usepackage{authblk}

\author[1,2]{Shengyuan Huang$^\dag$}
\author[3]{Wenjun Mei}
\author[1,2]{Xiaoguang Yang}
\author[4]{\hspace{1cm}Zhigang Cao$^*$}

\affil[1]{Academy of Mathematics and Systems Science, Chinese Academy of Sciences}
\affil[2]{University of Chinese Academy of Sciences}
\affil[3]{Peking University}
\affil[4]{Beijing Jiaotong University}

\begin{document}

\maketitle
\renewcommand{\thefootnote}{*}
\footnotetext{Corresponding Author,  Email: \href{zgcao@bjtu.edu.cn}{zgcao@bjtu.edu.cn}.}

\begin{abstract}
This paper studies allocation mechanisms in max-flow games with players' capacities as private information. 
We first show that no core-selection mechanism is truthful: there may exist a player whose payoff increases if she under-reports her capacity when a core-section mechanism is adopted. 
We then introduce five desirable properties for mechanisms in max-flow games: DSIC (truthful reporting is a dominant strategy), SIR (individual rationality and positive payoff for each player contributing positively to at least one coalition), SP (no edge has an incentive to split into parallel edges), MP (no parallel edges have incentives to merge), and CM (a player's payoff does not decrease as another player's capacity and max-flow increase). While the Shapley value mechanism satisfies DSIC and SIR, it fails to meet SP, MP and CM. We propose a new mechanism based on minimal cuts that satisfies all five properties.
\end{abstract}

\section{Introduction}

Given a directed network with source nodes and sink nodes as well as edge capacities, the max-flow problem asks to  transport as much flow  from source nodes to sink nodes as possible subject to capacity constraints.
This problem is a classical  one in the theories of optimization and networks, with numerous applications \citep{Ford_Fulkerson_1956,shapley1961network,gale1981C/S,1982kalaibalanced,Kalai1982,candogan2021competitive,sargent2024economic}. 

Suppose now each edge is owned by a player, and players may cooperate to maximize the transported flow by using only the edges controlled by them. It can be observed that the most efficient outcome is when all players cooperate to obtain the maximum flow. An interesting question is how to fairly allocate the obtained revenue, assuming that each unit of transported flow brings a unit revenue. This revenue allocation problem in max-flow has been extensively studied in the literature,  from a cooperative game theory perspective mainly using the core and the nucleolus \citep{1982kalaibalanced,Kalai1982,deng2008algorithmic,Deng2009nucleolus,2023vvv,vvv2024}.

Most relevant papers assume that the capacities of the edges are all public information and that players contribute all their capacities when cooperating. However, neither of these above assumptions is satisfactory: edge capacities may well be private information, and even if they are public information, players may decide to contribute only part of their capacities when cooperating.

This paper studies revenue allocation in max-flow games with players' capacities as private information. We aim to design allocation mechanisms to solicit their true information, and then allocate the obtained total revenue fairly.

\subsection{Main results}
We first consider core allocations, because they have been extensively studied for the max-flow game.  
We find that no core-selection mechanism is truthful: there may exist a player whose payoff increases if she under-reports her capacity when a core-selection mechanism is adopted, even if all other players are truth-telling (Example \ref{eg:core1}).

We then introduce five desirable properties for mechanisms in max-flow games: DSIC (truthful reporting is a dominant strategy), SIR (individual rationality and a positive payoff for each player contributing positively to at least one coalition), SP (no edge has an incentive to split into parallel edges), MP (no parallel edges have incentives to merge), and CM (a player's payoff does not decrease as another player's capacity and max-flow increase). 

While the Shapley value mechanism satisfies DSIC and SIR (Proposition \ref{Monotonicity of Shapley value}), it fails to meet SP, MP, and CM  (Examples \ref{eg:sh1},\ref{eg:sh2},\ref{eg:sh3}). We also study CM of the Shapley value mechanism in more depth from the perspective of complements and substitutes explored by \cite{shapley1962complements} (Proposition \ref{shcom}).

We propose a new mechanism based on minimal cuts, referred to as the MC mechanism. After all players report their capacities, the MC mechanism proceeds in two steps.  First, for players who can transport some flow on their own, we allocate to them their stand-alone values. Second, we remove all players in the first step and allocate the remaining value equally among all minimal cuts. Then, for each minimal cut, every involved player receives a payoff that is proportional to her capacity. Our major result is that the MC mechanism satisfies all five properties (Theorem \ref{Monotonicity of MC}). One advantage of the  MC mechanism is that it uses minimal cuts instead of merely minimum cuts: where the  minimal cut is a purely structural concept, the minimum cut is not (it depends on the capacities of the edges).

We also study the CM property of the MC mechanism in more depth, and provide a more accurate and complete analysis of the impact on other players when the reported capacity of one player increases  (Theorem \ref{Cross-effect between edges}). A comparison of the major mechanisms discussed in this paper is provided in Table \ref{tab:your_label}.

\begin{table}[h]
\centering
\caption{A comparison of major mechanisms}
\label{tab:your_label} 
\begin{tabular}{cccccc}
\toprule
 & DSIC & SIR & SP & MP & CM \\
\midrule
Core &  × & × & \checkmark & \checkmark & ×\\
Shapley Value & \checkmark & \checkmark & × & × & × \\
MC mechanism& \checkmark & \checkmark & \checkmark & \checkmark & \checkmark\\
\bottomrule
\end{tabular}
\end{table}

\subsection{Literature review}
\cite{1982kalaibalanced,Kalai1982} were the first to formally define max-flow games from the perspective of cooperative game theory. They demonstrated that the core of a flow game has a relatively explicit structure and established that flow games are equivalent to totally balanced games, implying that every subgame of a flow game has a non-empty core. The more general multi-commodity flow games were considered in \cite{derks1985stable}.  Most follow-up studies on max-flow games  focused on the computation and refinement of the core and the nucleolus. For instance, \cite{potters2006nucleolus} presented an algorithm for computing the nucleolus of simple flow games in a general setting, \cite{Deng2009nucleolus} showed that the nucleolus of a simple flow game can be efficiently computed in polynomial time using a linear programming duality approach, \cite{Kern2009morflow} developed algorithms for computing the core and f-nucleon of simple flow games.  

Regarding core refinement, \cite{2023vvv} employed the dual linear program of the max-flow linear program to partially characterize the core,  \cite{vvv2024} introduced the Dual-Consistent Core (DCC) as a refinement of the core,  offering improved fairness and computational efficiency compared to the core.

While edges in the max-flow problem are viewed  players in the aforementioned literature (including this paper), there are also models that view nodes as players.  Investigating Internet traffic flows using  cooperative game theory is also one of the earliest models in the community of algorithmic game theory, where each node of the network is a server and is viewed as a player  in the model \citep{papadimitriou2001algorithms,markakis2003core}. \cite{mor2024} contributed to understanding cooperative behaviors in decentralized flow networks by exploring flow allocation games, where players controlling nodes allocate flows strategically. They analyzed the computational complexity of equilibrium solutions, characterized core allocations, and proposed algorithms to ensure stability in flow distribution.

To the best of our knowledge, \cite{archer2001truthful} were the first to consider mechanism design problems in the context of max-flow, as an application of the general theory they studied with a single private parameter for each player. In the setting of max-flows,  \cite{archer2001truthful} assumed that the single private parameter of each player is (eventually) the edge capacity, as in this paper. However, their payoff setting differs significantly from ours: they used the inverse of the edge capacity to define the unit transportation cost, and the payoff for each player is the allocated payoff minus the total cost (the inverse of this edge capacity times the flow amount). Moreover, they also assumed that the transferred payoff  and the allocated flow of each player (given other players' reported parameters) are twice differentiable in the player's reported parameter, which is not the case in our paper (they are not even differentiable). 
Mechanism design problems related to max-flows are also considered in \cite{lavi2011truthful} and \cite{liu2014mechanism}, where players are neither edges nor nodes, but source-sink pairs (thus their models are closer to the congestion game). 
All mentioned papers concentrated on truthfulness of the mechanism as well as computational issues, whereas we do not focus on computation but on more desirable properties of the mechanism. To the best of our knowledge, we are the first to consider SIR, SP, MP, and CM in the setting of max-flows.
\cite{agarwal2008mechanism} also considered mechanism design in max-flows problems and each player in their model may possess capacities of multiple edges. However, they did not consider truthfulness issues and only core allocations were studied.  

\section{Preliminaries}

\subsection{Model}





    \paragraph{Max-flow problem.} We are given an acyclic directed graph \(G=(V,E,c)\), where $V$ is the set of nodes, $E$ is the set of directed edges, and $c\in \mathbb{R}_{++}^E$ is the vector of capacities, with $c(e_i)=c_i>0$ being the capacity of each edge $e_i\in E$.
    A node is a source of \(G\) if its in-degree is 0 and out-degree is positive, and a sink if its out-degree is 0 and in-degree is positive.  We use \(V_s\) to denote the set of sources, and \(V_t\) the set of sinks. We assume  (except for the extension model) that $G$ has a single source denoted $s$ and a single sink denoted $t$, and assume that each edge is on at least one $s$-$t$  path.
    
    A flow is a map \(f : E \to \mathbb{R}_{+}\) that satisfies (i) Capacity constraint: \( 0\leq f(e) \leq c(e), \forall e \in E\), and (ii) Conservation of flows: \(\sum_{e\in \delta^-(v)} f(e) = \sum_{e\in \delta^+(v)} f(e),\forall v \in V \setminus (V_s\cup V_t),\)  where \(\delta^-(v)\) (resp. \(\delta^+(v))\) denotes the set of out-edges (resp. in-edges)  of $v$.
The max-flow problem is to find a flow $f$ that maximizes \(\sum_{v \in V_s}\sum_{e\in \delta^-(v)} f(e)\). We use $F(c)$ to denote the max-flow value when the capacities are $c$.

 \paragraph{Max-flow game.} Suppose that one unit of profit is earned for each unit of flow, and hence the flow value can be viewed  as the total payoff.
  Suppose also that every edge $e_i\in E$ is possessed by a player $i$, which may also be denoted $e$. Then each max-flow problem is associated with a cooperative game \((N,v)\), where \(N=E\) is the player set,  and \(v: 2^N \rightarrow \mathbb{R}\) is the characteristic function such that for any coalition \(S \subseteq N\), \(v(S)\) is the max-flow value using edges in $S$. We often suppose that player \(i \in N\) owns \(e_i \in E\), we denote the capacity \(c(e_i)\) by \(c_i\) for simplicity.

For each vector \(x\in \mathbb{R}^N\), we denote \(x(S) =\sum_{i\in S}x(i)\). The core of cooperative game \((N,v)\) is a set of efficient allocations such that no group of players have incentives to deviate, i.e., 
\(C(N,v)=\{x\in \mathbb{R}^N| ~x(S)\geq v(S), \forall S\subset N; x(N)=v(N) \}.\)
The max-flow game has been extensively studied, and it is known that the max-flow game has a nonempty core \citep{1982kalaibalanced}.



\paragraph{Mechanism design.} The model of the max-flow game assumes that all the information of the games, especially the capacities of the edges, are known to the social planner. However, this assumption is often unrealistic.
We suppose throughout the paper that each capacity $c(e)$ is private information.

Since $c(e)$ is private information, player $i$ may have incentive to misreport it. Let $\hat c(e)$ be the reported capacity of $e$, we assume that  $\hat c(e)\leq c(e)$. That is, each player can under-report her capacity, but cannot over-report. This is an assumption that is frequently made in the literature. Denote $(N,\hat v)$ by the max-flow game induced from the reporting capacity function $\hat c$.
 
Roughly speaking, a mechanism of a max-flow problem is a payoff allocation based on reported capacities. 
We denote a mechanism as a map $\phi$ from reported capacity vectors \(\hat c\in \mathbb{R}^n_{++}\) to an allocation vector \(\phi(\hat c)=(\phi_1(\hat c), \phi_2(\hat c)),..., \phi_n(\hat c))\), such that $\sum_{i=1}^n \phi_i(\hat c) =\hat v(N)$. That is, we assume that each mechanism for the max-flow problem is efficient.

Besides payoff allocation, we also assume that the designer selects a set of paths to transport the maximum flow. Since there is no cost incurred when the edge of a player is used,  and generally the edge is useless to any single player when it is not selected, the problem of path selection is unimportant and will not be explicitly addressed. 

The first standard desirable property of a mechanism is that no player has an incentive to lie. There are two versions of this property, depending on whether Nash equilibrium  or dominant strategy equilibrium is applied.

\begin{itemize}
\item[-] {\bf Nash Equilibrium Incentive Compatibility (NEIC):}
A mechanism  \(\phi:\mathbb{R}_{++}^n\rightarrow \mathbb{R}^n\) is NEIC, if truthful reporting by all players forms a Nash equilibrium. That is, for each player $i\in N$: 
\(\phi_i(c_i,c_{-i}) \geq \phi_i(\hat c_i,c_{-i}), \forall \hat c_i\in [0,c_i].\)
\item[-] {\bf Dominant Strategy Incentive Compatibility (DSIC):}
A mechanism  \(\phi:\mathbb{R}_{++}^n\rightarrow \mathbb{R}^n\) is DSIC, if truthful reporting is a dominant strategy for every player. That is, for each player $i\in N$: 
\(\phi_i(c_i,\hat c_{-i}) \geq \phi_i(\hat c_i,\hat c_{-i}), \forall \hat c_i\in [0,c_i], \hat c_{-i} \in [0,c_{-i}].\)
\end{itemize}

The second standard desirable property of a mechanism is that players have incentives to participate in the game, i.e., the allocated payoff is no less than her stand-alone value. This property is known as individual rationality (IR). We shall apply a stronger version of IR, requiring both IR and that each player generally receives a positive payoff (i.e., each player has an absolute incentive to participate in it).  

\begin{itemize}
\item[-] {\bf Strong Individual Rationality (SIR)}: 
A mechanism  \(\phi:\mathbb{R}_{++}^n\rightarrow \mathbb{R}^n\) is SIR, if  (i) for each player $i\in N$,
\(\phi_i(\hat c) \geq \hat v_i=\hat v(e_i)\), and (ii) for each player $i$ if there exists at least one coalition \(S \subset E\) meeting \(\hat v(S)-\hat v(S\setminus \{e_i\})>0\), we have
\(\phi_i(\hat c) >0.\) Equivalently, a mechanism is SIR, if it is IR, and allocates a positive payoff to any player whose edge is on at least one source-terminal path.
\end{itemize}

\subsection{Core-selection mechanisms}

Since the core of the max-flow game has been extensively studied, and core-selection is a desirable property,  it is natural to ask whether there exists any core-selection mechanism that is DSIC. The example below provides a negative answer, even to NEIC.

\begin{example}[Core violates NEIC and SIR] \label{eg:core1}Consider any core-selection mechanism $\phi(\hat c)$ and a four-player game in Figure \ref{counterexample}. When $\hat c=(\hat c_1, \hat c_2,\hat c_3,\hat c_4)=(1,1,1,1)$, since $\sum_{i=1}^n \phi_i(\hat c) = \hat v(N)=2$, there must exist one player with a positive payoff. By symmetry, we assume that player 1 has $\phi_1(\hat c)>0$ without loss of generality. 

When $c=(c_1,c_2,c_3,c_4)=(2,1,1,1)$, suppose that players 2,3, and 4 report their true capacities of 1. If player 1 reports a true capacity of 2, the unique core allocation will be $\phi(\hat c')=(\phi_1(\hat c'),\phi_2(\hat c'),\phi_3(\hat c'),\phi_4(\hat c'))=(0,0,1,1)$, in which case the profit allocated to player 1 is 0. This implies that the core violates SIR. Furthermore, if player 1 strategically under-reports a capacity as 1, then player 1's core allocation will be $\phi_1(1,1,1,1)>0$. Therefore, player 1 has an incentive to lie, and the mechanism cannot be NEIC. 
\end{example}

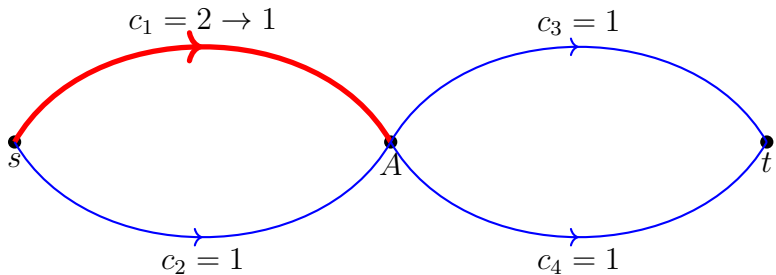
\begin{figure}[H]
 \centering
\begin{tikzpicture}[scale=5]

    \fill (1,0) circle (0.5pt) node[below] {$s$};
    \fill (2,0) circle (0.5pt) node[below] {$A$};
    \fill (3,0) circle (0.5pt) node[below] {$t$};

  \draw[red, line width=2pt, postaction={decorate,decoration={markings,mark=at position 0.5 with {\arrow{>[scale=1]}}}}] (1,0) to[out=60, in=120] node[above] {\textcolor{black}{$c_1=2 \rightarrow 1$}} (2,0);
 \draw[blue, line width=0.8pt, postaction={decorate,decoration={markings,mark=at position 0.5 with {\arrow{>[scale=1]}}}}](1,0) to[out=-60, in=-120] node[below] {\textcolor{black}{$c_2=1$}} (2,0);
  \draw[blue, line width=0.8pt, postaction={decorate,decoration={markings,mark=at position 0.5 with {\arrow{>[scale=1]}}}}] (2,0) to[out=60, in=120] node[above] {\textcolor{black}{$c_3=1$}} (3,0);
   \draw[blue, line width=0.8pt, postaction={decorate,decoration={markings,mark=at position 0.5 with {\arrow{>[scale=1]}}}}] (2,0) to[out=-60, in=-120] node[below] {\textcolor{black}{$c_4=1$}} (3,0);
\end{tikzpicture}

\caption{Core violates NEIC and SIR}
\label{counterexample}
\end{figure}

\begin{remark}This example could be viewed as a Braess-like Paradox. Suppose now some core-selection mechanism is applied. Then it may well be the case that a player's payoff decreases when her capacity increases. Note also the difference between this example and the typical Braess Paradox: while the latter focuses on the overall congestion time, this example focuses on the individual payoff. Regarding network structures, this example is even stronger than the Braess Paradox, because it is known that the Braess Paradox does not happen for series-parallel networks \citep{milchtaich2006network,chen2016network,Kern2009morflow}. However, series-parallel networks indeed have additional desirable properties in our setting, as will be seen later. 
\end{remark}

\begin{remark}The nucleolus is another solution concept for max-flow game that has been extensively studied \citep{potters2006nucleolus,Deng2009nucleolus}. Since the nucleolus is a subset of the core whenever the latter is nonempty, and the core of a max-flow game is non-empty,  we know that there does not exist a nucleolus-selection mechanism that is NEIC for all max-flow games.\end{remark}

\section{The Shapley Value Mechanism}

We denote the reported capacity of player $i\in N$ by $\hat c_i \in [0, c_i]$, which leads to a new max-flow game \((N,\hat v)\). The Shapley value mechanism allocates to each player $i$ its Shapley value in \((N,\hat v)\) as the payoff, which is the expected marginal contribution of \(i\), i.e.,
\[Sh_i(\hat c) = \sum_{i\in S, S \subseteq N} \frac{(|S|-1)! (|N| - |S|)!}{|N|!}\left( \hat v(S) - \hat v(S \setminus \{i\} ) \right).\]

\begin{proposition}
The Shapley value mechanism is DSIC and SIR for max-flow games.
\label{Monotonicity of Shapley value}
\end{proposition}

\begin{proof}
Consider any player \(i\) owning edge \(e_i\).  Since \(\hat v(S)\) is weakly increasing  in \(\hat c_i\) for any $S$ containing $e_i$, it is evident that the Shapley value mechanism is DSIC. Since \(\hat v(S)-\hat v(S\setminus i)\) is always nonnegative, to show that the Shapley value mechanism is IR, we only need to consider the case that $\hat v(\{i\})>0$. In this case, $e_i$ must be an $s-t$ edge, and hence  \(\hat v(S)-\hat v(S\setminus i)=\hat v(\{i\})\) for all $S\ni i$, implying that \(Sh_i(\hat c)=\hat v(\{i\})\). Thus the Shapley value mechanism is IR. To show SIR, suppose now there exists a coalition \(S\) such that \(v(S)-v(S\setminus i)>0\). By definition, we immediately have \(Sh_i(v)>0\), proving that the Shapley value mechanism is SIR.
\end{proof}

\begin{remark}The above simple property regarding DSIC could be better understood from a general property of the Shapley value. A solution concept $\phi$ is called strongly monotonic, if $\phi_i(v)\leq \phi_i(v')$ holds for any two cooperative games $(N,v)$ and $(N,v')$ and a player $i\in N$ such that $v_i(S)-v_i(S\setminus \{i\})\leq v'_i(S)-v'_i(S\setminus \{i\})$ for all $S\subseteq N, i\in S$. For the class of all cooperative games, it is known that the Shapley value is the unique symmetric and efficient solution concept that is strongly monotonic \citep{Young1985}.  For general mechanism design problems, truth-telling is also closely related to monotonicity \citep{ashlagi2010monotonicity}.
\end{remark}

\subsection{Violations of more desirable properties}

Beyond DSIC and SIR, additional desirable properties can be conceived for mechanisms designed to tackle the payoff-allocation  problem in  max-flow.


\begin{itemize}
\item[-] {\bf Split-Proofness (SP):}
A mechanism  \(\phi:\mathbb{R}_{++}^{|E|} \rightarrow \mathbb{R}^{|E|}\) is SP, if any edge has no incentive to split. That is, for any edge \(e_i\), it splits into two basic parallel edges \(e_{i1}, e_{i2}\) (Figure \ref{S-I}) with $ c_{i}=c_{i1}'+ c_{i2}'$ and $\hat c_{i}=\hat c_{i1}'+\hat c_{i2}', \forall c'_{i1}, c'_{i2},\hat c'_{i1}, \hat c'_{i2} \in R_{+}$, we always have \(\phi_i(\hat c) \geq \phi_{i1}(\hat c')+\phi_{i2}(\hat c')\).

\item[-]{\bf Merge-Proofness (MP):}
A mechanism  \(\phi:\mathbb{R}_{++}^{|E|} \rightarrow \mathbb{R}^{|E|}\) is MP, if any two basic parallel edges have no incentive to merge. That is, for any two basic parallel edges \(e_i, e_j \in E\) (Figure \ref{S-I}), they merge into an edge $e_{i,j}$ with $ c_{i}+ c_{j}=c_{i,j}'$ and $ \hat c_{i}+\hat c_{j}=\hat c_{i,j}'$, we have \(\phi_{i}(\hat c)+\phi_{j}(\hat c) \geq \phi_{i,j}(\hat c') \).

\item[-]{\bf Cross Monotonicity (CM)}: 
A mechanism  \(\phi:\mathbb{R}_{++}^{|E|} \rightarrow \mathbb{R}^{|E|}\) is CM, if an increase in the reported capacity of an edge leads to a strict increase of the maxi-flow value and a weak increase of the payoffs of all other edges. That is, for any edge $e_i$, if its reported capacity $\hat c_i$ strictly increases to $\hat c_i'$, we have \(\phi_{j}(\hat c)\leq \phi_{j}(\hat c'), \forall j \neq i\).

\end{itemize}

\begin{example}[Shapley value violates SP] \label{eg:sh1}
Consider the 6-player game described in  Figure \ref{The Shapley value is not SP}(a), where player 1 has a capacity of 2 and all other players have capacities of 1. When each player truthfully reports her capacity, it can be verified that $Sh_1(\hat c)$ is \(\frac{1}{30}\). Suppose now player 1 splits her edge into two parallel ones with unit capacity, as described in Figure \ref{The Shapley value is not SP}(b). 
It can be verified that each of the new players has a Shapley value of \(\frac{1}{42}\). Since  \(\frac{2}{42}>\frac{1}{30}\), player 1 does have an incentive to do this, demonstrating that the Shapley value violates SP.

\end{example}

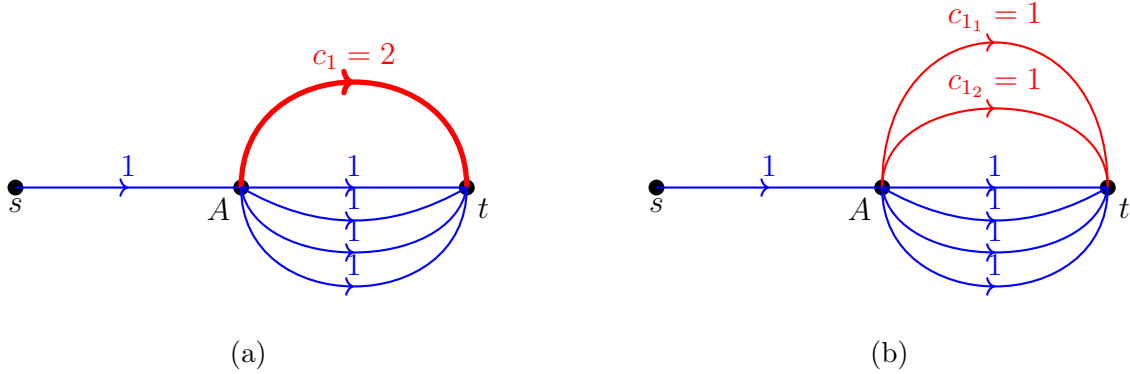
\begin{figure}[H] 

\begin{subfigure}[b]{0.45\textwidth}
\centering
\begin{tikzpicture}[scale=3]
      \fill (1,0) circle (1pt) node[below] {$s$};
   \fill (2,0) circle (1pt) node[below left] {$A$};
   \fill (3,0) circle (1pt) node[below right] {$t$};

  \draw[blue, line width=0.8pt, postaction={decorate,decoration={markings,mark=at position 0.5 with {\arrow{>[scale=1]}}}}] (1,0) to[out=0, in=180] node[above] {$1$} (2,0);
  \draw[red, line width=2pt, postaction={decorate,decoration={markings,mark=at position 0.5 with {\arrow{>[scale=0.8]}}}}] (2,0) to[out=90, in=90, looseness=1.6] node[above] {$c_{1}=2$} (3,0);
  
  \draw[blue,line width=0.8pt, postaction={decorate,decoration={markings,mark=at position 0.5 with {\arrow{>[scale=1]}}}}] (2,0) to[out=-0, in=-180] node[above] {$1$} (3,0);
  \draw[blue,line width=0.8pt, postaction={decorate,decoration={markings,mark=at position 0.5 with {\arrow{>[scale=1]}}}}] (2,0) to[out=-30, in=-150] node[above] {$1$} (3,0);
  \draw[blue,line width=0.8pt, postaction={decorate,decoration={markings,mark=at position 0.5 with {\arrow{>[scale=1]}}}}] (2,0) to[out=-80, in=-100] node[above] {$1$} (3,0);
  \draw[blue,line width=0.8pt,postaction={decorate,decoration={markings,mark=at position 0.5 with {\arrow{>[scale=1]}}}}] (2,0) to[out=-90, in=-90, looseness=1.5] node[above] {$1$} (3,0);
\end{tikzpicture}
\caption{}
\label{Edge merging2}
\end{subfigure}
\hfill
\begin{subfigure}[b]{0.45\textwidth}

\centering

\begin{tikzpicture}[scale=3]

   \fill (1,0) circle (1pt) node[below] {$s$};
   \fill (2,0) circle (1pt) node[below left] {$A$};
   \fill (3,0) circle (1pt) node[below right] {$t$};

  \draw[blue,line width=0.8pt,postaction={decorate,decoration={markings,mark=at position 0.5 with {\arrow{>[scale=1]}}}}] (1,0) to[out=0, in=180] node[above] {$1$} (2,0);
  \draw[red,line width=0.8pt,postaction={decorate,decoration={markings,mark=at position 0.5 with {\arrow{>[scale=1]}}}}] (2,0) to[out=90, in=90, looseness=2.2] node[above] {$c_{1_1}=1$} (3,0);
  \draw[red,line width=0.8pt,postaction={decorate,decoration={markings,mark=at position 0.5 with {\arrow{>[scale=1]}}}}] (2,0) to[out=90, in=90, looseness=1.2] node[above] {$c_{1_2}=1$} (3,0);
  
  \draw[blue,line width=0.8pt,postaction={decorate,decoration={markings,mark=at position 0.5 with {\arrow{>[scale=1]}}}}] (2,0) to[out=-0, in=-180] node[above] {$1$} (3,0);
  \draw[blue,line width=0.8pt,postaction={decorate,decoration={markings,mark=at position 0.5 with {\arrow{>[scale=1]}}}}] (2,0) to[out=-30, in=-150] node[above] {$1$} (3,0);
  \draw[blue,line width=0.8pt,postaction={decorate,decoration={markings,mark=at position 0.5 with {\arrow{>[scale=1]}}}}] (2,0) to[out=-80, in=-100] node[above] {$1$} (3,0);
  \draw[blue,line width=0.8pt,postaction={decorate,decoration={markings,mark=at position 0.5 with {\arrow{>[scale=1]}}}}] (2,0) to[out=-90, in=-90, looseness=1.5] node[above] {$1$} (3,0);
\end{tikzpicture}
\caption{}
\label{Edge splitting2}
\end{subfigure}

\caption{Shapley value violates SP}
\label{The Shapley value is not SP}
\end{figure}

\begin{example}[Shapley value violates MP] \label{eg:sh2} Consider the 3-player game described in Figure \ref{Non-satisfaction of anti-sibil attack by the Shapley value}(a) with all unit capacities. When each player truthfully reports her capacity, it can be verified that  $Sh_1(\hat c), Sh_2(\hat c)$ are both \(\frac{1}{6}\). Suppose now the two players merge into a new one with capacity 2, as described in Figure \ref{Non-satisfaction of anti-sibil attack by the Shapley value}(b). It can be verified that  the Shapley value of the new player is \(\frac{1}{2}\) . Since  $\frac{2}{6}<\frac{1}{2}$, the two players do have incentives to do this, demonstrating that the Shapley value violates MP.\end{example}

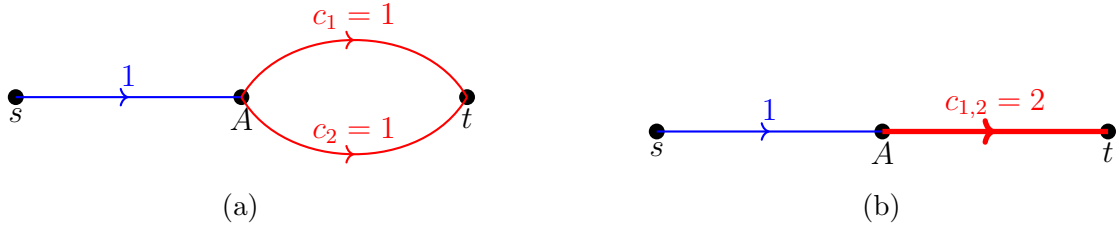
\begin{figure}[h] 
\begin{subfigure}[b]{0.45\textwidth}
\centering
\begin{tikzpicture}[scale=3]

   \fill (1,0) circle (1pt) node[below] {$s$};
   \fill (2,0) circle (1pt) node[below] {$A$};
   \fill (3,0) circle (1pt) node[below] {$t$};

  \draw[blue,line width=0.8pt, postaction={decorate,decoration={markings,mark=at position 0.5 with {\arrow{>[scale=1]}}}}] (1,0) to[out=0, in=180] node[above] {$1$} (2,0);
  \draw[red,line width=0.8pt, postaction={decorate,decoration={markings,mark=at position 0.5 with {\arrow{>[scale=1]}}}}] (2,0) to[out=60, in=120] node[above] {$c_1=1$} (3,0);
  \draw[red,line width=0.8pt, postaction={decorate,decoration={markings,mark=at position 0.5 with {\arrow{>[scale=1]}}}}] (2,0) to[out=-60, in=-120] node[above] {$c_2=1$} (3,0);
\end{tikzpicture}
\caption{}
\label{Edge splitting}
\end{subfigure}
\hfill
\begin{subfigure}[b]{0.45\textwidth}
\centering
\begin{tikzpicture}[scale=3]
   \fill (1,0) circle (1pt) node[below] {$s$};
   \fill (2,0) circle (1pt) node[below] {$A$};
   \fill (3,0) circle (1pt) node[below] {$t$};
  \draw[blue,line width=0.8pt, postaction={decorate,decoration={markings,mark=at position 0.5 with {\arrow{>[scale=1]}}}}] (1,0) to[out=0, in=180] node[above] {$1$} (2,0);
  \draw[red, line width=2pt, postaction={decorate,decoration={markings,mark=at position 0.5 with {\arrow{>[scale=0.8]}}}}] (2,0) to[out=0, in=180] node[above] {$ c_{1,2}=2$} (3,0);
\end{tikzpicture}
\caption{}
\label{Edge merging}
\end{subfigure}

\caption{Shapley value violates MP}
\label{Non-satisfaction of anti-sibil attack by the Shapley value}
\end{figure}

\begin{example}[Shapley value violates CM] \label{eg:sh3} Consider the 4-player game described in Figure \ref{counterexample2}, where both left players have capacities of 0.5 and right players have capacities of 1. When each player truthfully reports her capacity, it can be verified that the max-flow value is 1.5, and $Sh_2(\hat c)$ is 0.5. 
Suppose now player 1$'$s  capacity increases from 0.5 to 0.6. Then the max-flow value increases from 1 to 1.1. However, $Sh_2(\hat c)$ decreases to 0.475. 
\end{example}

\begin{figure}[h]
 \centering
\begin{tikzpicture}[scale=5]
  \fill (1,0) circle (0.5pt) node[below] {$s$};
    \fill (2,0) circle (0.5pt) node[below] {$A$};
    \fill (3,0) circle (0.5pt) node[below] {$t$};

  \draw[red, line width=0.8pt, postaction={decorate,decoration={markings,mark=at position 0.5 with {\arrow{>[scale=1]}}}}] (1,0) to[out=60, in=120] node[above] {\textcolor{black}{$c_1=0.5 \rightarrow 0.6$}} (2,0);
 \draw[blue, line width=0.8pt, postaction={decorate,decoration={markings,mark=at position 0.5 with {\arrow{>[scale=1]}}}}](1,0) to[out=-60, in=-120] node[below] {\textcolor{black}{$c_2=0.5$}} (2,0);
  \draw[blue, line width=0.8pt, postaction={decorate,decoration={markings,mark=at position 0.5 with {\arrow{>[scale=1]}}}}] (2,0) to[out=60, in=120] node[above] {\textcolor{black}{$c_3=1$}} (3,0);
   \draw[blue, line width=0.8pt, postaction={decorate,decoration={markings,mark=at position 0.5 with {\arrow{>[scale=1]}}}}] (2,0) to[out=-60, in=-120] node[below] {\textcolor{black}{$c_4=1$}} (3,0);
\end{tikzpicture}

\caption{Shapley value violates CM}
\label{counterexample2}
\end{figure}
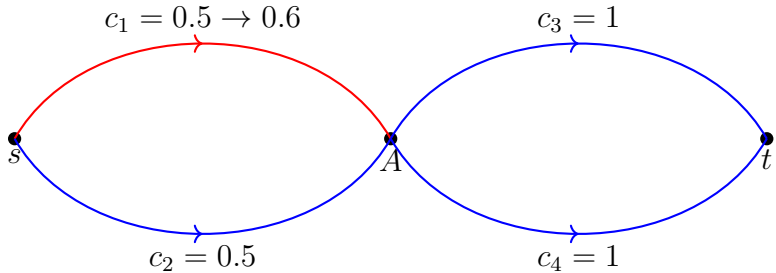

\subsection{More discussions on CM}
In this subsection, we discuss in more depth CM of the Shapley value mechanism. We shall show that, when one player's reported capacity increases, leading to an increase of the max-flow value, another player may always have her Shapley value increase, and may  always decrease, depending on their relationship. Before presenting this result, we introduce the concepts of complementarity and substitutability pioneered by \cite{shapley1961network}.

      A pair of edges \( e_i, e_j \in E\) is complementary (substitutable), if the following discrete second-order difference quotient is always non-negative (non-positive) for any \(x, y, a, b \in \mathbb{R}_+, ab\neq 0\):
    \[Q_{ij}(x,y,a,b)=\frac{F_{ij}(x+a,y+b)-F_{ij}(x+a,y)-F_{ij}(x,y+b)+F_{ij}(x,y)}{ab},\]
    where \(F_{ij}(x,y)\) is the maximum flow value when the capacities of $i$ and $j$ are $x$ and $y$, respectively, and the capacities of other edges are as determined in $c_{-ij}$.

    A surprising result proved in \cite{shapley1961network} is that each pair of edges are either complementary or substitutable. It is worth noting that this result depends on the capacities of other edges. To be more specific, when the capacities of other edges change, the relationship between two edges may also change. We strengthen these concepts as follows.



\begin{definition}[Constant Complementarity/Substitutability]
    A pair of edges \( e_i, e_j\) is constant complementary (substitutable), if  they are complementary (substitutable) for any configurations of the capacities of other edges.  
\end{definition}

Several types of  complementarity/substitutability investigated in \cite{shapley1961network} are actually constant, as illustrated in Figure \ref{Constant supermodularity-type} and Figure \ref{Constant submodularity-type}.  Constant complementarity and substitutability are useful concepts for us to investigate CM of the Shapley value.  We shall show below that CM is valid if two players are constantly complementary, and invalid if they are constantly substitutable.


\begin{figure}[H]
    \centering
    \begin{subfigure}[b]{0.45\textwidth}
        \centering
        \begin{tikzpicture}[scale=2.5]
            \fill (1,0) circle (1pt) node[below] {$1$};
            \fill (2,0) circle (1pt) node[below] {$2$};
            \fill (3,0) circle (1pt) node[below] {$3$};
            \draw[line width=0.8pt, postaction={decorate,decoration={markings,mark=at position 0.5 with {\arrow{>[scale=1]}}}}] (1,0) to[out=0, in=180] node[above] {$e_1$} (2,0);
            \draw[line width=0.8pt, postaction={decorate,decoration={markings,mark=at position 0.5 with {\arrow{>[scale=1]}}}}] (2,0) to[out=0, in=180] node[above] {$e_2$} (3,0);
        \end{tikzpicture}
        \caption{C-I}
        \label{C-I}
    \end{subfigure}
    \hfill
    \begin{subfigure}[b]{0.45\textwidth}
        \centering
        \begin{tikzpicture}[scale=2.5]
            \fill (1,0) circle (1pt) node[left] {$source$};
            \fill (2,0.5) circle (1pt) node[above] {$1$};
            \fill (2,-0.5) circle (1pt) node[below] {$2$};
            \fill (3,0) circle (1pt) node[right] {$sink$};
            \draw[line width=0.8pt, postaction={decorate,decoration={markings,mark=at position 0.5 with {\arrow{>[scale=1]}}}}] (1,0) to[out=60, in=180] node[above] {$e_1$} (2,0.5);
            \draw[line width=0.8pt, postaction={decorate,decoration={markings,mark=at position 0.5 with {\arrow{>[scale=1]}}}}] (2,-0.5) to[out=0, in=-120] node[below] {$e_2$} (3,0);
        \end{tikzpicture}
        \caption{C-II}
        \label{C-II}
    \end{subfigure}
    \caption{Two types of constant complementarity}
    \label{Constant supermodularity-type}
\end{figure}
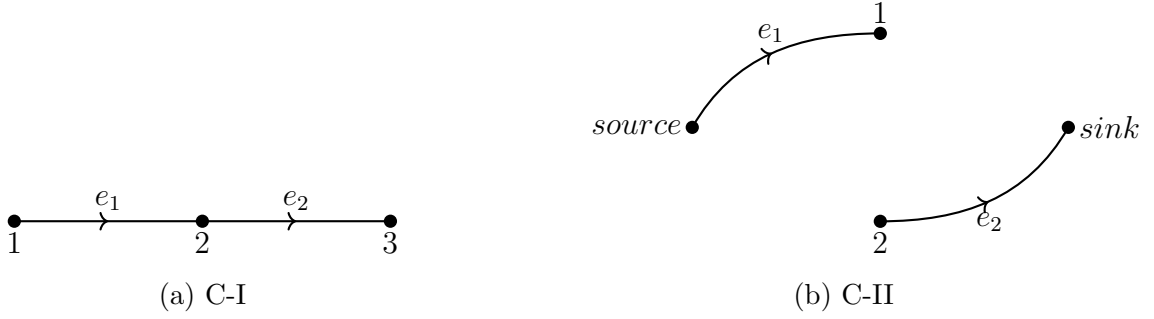

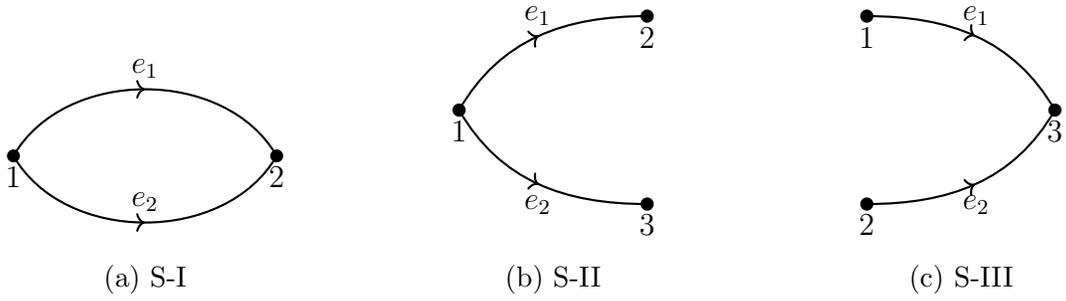
\begin{figure}[h]
    \centering
    \begin{subfigure}[b]{0.3\textwidth}
        \centering
        \begin{tikzpicture}[scale=3.5]
            \fill (2,0) circle (0.7pt) node[below] {$1$};
            \fill (3,0) circle (0.7pt) node[below] {$2$};
            \draw[line width=0.8pt, postaction={decorate,decoration={markings,mark=at position 0.5 with {\arrow{>[scale=1]}}}}] (2,0) to[out=60, in=120] node[above] {$e_1$} (3,0);
            \draw[line width=0.8pt, postaction={decorate,decoration={markings,mark=at position 0.5 with {\arrow{>[scale=1]}}}}] (2,0) to[out=-60, in=-120] node[above] {$e_2$} (3,0);
        \end{tikzpicture}
        \caption{S-I}
        \label{S-I}
    \end{subfigure}
    \hfill
    \begin{subfigure}[b]{0.3\textwidth}
        \centering
        \begin{tikzpicture}[scale=2.5]
            \fill (2,0) circle (1pt) node[below] {$1$};
            \fill (3,0.5) circle (1pt) node[below] {$2$};
            \fill (3,-0.5) circle (1pt) node[below] {3};
            \draw[line width=0.8pt, postaction={decorate,decoration={markings,mark=at position 0.5 with {\arrow{>[scale=1]}}}}] (2,0) to[out=60, in=180] node[above] {$e_1$} (3,0.5);
            \draw[line width=0.8pt, postaction={decorate,decoration={markings,mark=at position 0.5 with {\arrow{>[scale=1]}}}}] (2,0) to[out=-60, in=-180] node[below] {$e_2$} (3,-0.5);
        \end{tikzpicture}
        \caption{S-II}
        \label{S-II}
    \end{subfigure}
    \hfill
    \begin{subfigure}[b]{0.3\textwidth}
        \centering
        \begin{tikzpicture}[scale=2.5]
            \fill (1,0.5) circle (1pt) node[below] {$1$};
            \fill (1,-0.5) circle (1pt) node[below] {$2$};
            \fill (2,0) circle (1pt) node[below] {$3$};
            \draw[line width=0.8pt, postaction={decorate,decoration={markings,mark=at position 0.5 with {\arrow{>[scale=1]}}}}] (1,0.5) to[out=0, in=120] node[above] {$e_1$} (2,0);
            \draw[line width=0.8pt, postaction={decorate,decoration={markings,mark=at position 0.5 with {\arrow{>[scale=1]}}}}] (1,-0.5) to[out=0, in=-120] node[below] {$e_2$} (2,0);
        \end{tikzpicture}
        \caption{S-III}
        \label{S-III}
    \end{subfigure}
    \caption{Three types of constant substitutability}
    \label{Constant submodularity-type}
\end{figure}

\begin{proposition}\label{shcom}
If \( (e_i, e_j)\) is constantly complementary (constantly substitutable), then $Sh_j(\hat c)$ is increasing (decreasing) in \(\hat c_i\).
\end{proposition}

\begin{proof}
We only consider the case of constant complementarity, as the case of constant substitutability can be proved in a similar way. Suppose now the reported capacity of player $i$ increases from \(\hat c_i\) to \(\hat c'_i\).
For any coalition \(S\subseteq N\) such that \(\{i, j\} \subset S\), consider the subgraph \(G_S\) induced by \(S\). We use  \(\hat F^S_{i,j}(\cdot,\cdot)\) to denote the max-flow function of $i,j$ based on the reported capacities when all edge capacities outside of $S$ are zero.  Due to the constant complementarity of \( (e_i, e_j)\),
we have
 \[\hat F^S_{i,j}(\hat c'_i,\hat c_j)-\hat F^S_{i,j}(\hat c'_i,0)  \geq 
     \hat F^S_{i,j}(\hat c_i,\hat c_j)-\hat F^S_{i,j}(\hat c_i,0).\]

Let   \(\hat v(S|\hat c_i)\) denote the value of coalition \(S\) with the reported capacity of  player $i$ \(\hat c_i\), the above inequality can be rewritten as
\[\hat v(S|\hat c'_i)-v(S\setminus\{j\}|\hat c'_i) \geq \hat v(S|\hat c_i)-\hat v(S\setminus\{j\}|\hat c_i),\]
implying that \(\hat v(S|\hat c_i)-\hat v(S\setminus\{j\}|\hat c_i)\) increases in \(\hat c_i\). 
By definition, it can be seen that \(Sh_{j}(\hat c)\)  is increasing in \(\hat c_i\).
\end{proof}

Proposition \ref{shcom} could be viewed as a comparative statics of the Shapley value on the max-flow problem. It implies that the Shapley value mechanism for max-flow problems does satisfy CM if all pairs of players are constantly complementary (e.g., when the network consists of several separate $s$-$t$ paths).

\section{Minimal-cut Based Mechanism}

Since the Shapley value violates the three desirable properties SP, MP and CM, it is natural to ask whether we are able to design a mechanism such that all the listed desirable properties are met. In this section, we propose such a mechanism that is based on minimal-cuts. 

\subsection{Definition}
    
    Given an acyclic digraph \(G = (V,E,c)\), a set of edges \( W \subseteq E\) is called a cut if there is no $s$-$t$ path  when  all edges in $W$ are removed. A  cut $M$ is minimal if \(M\setminus \{e\}\) is not a cut  for any edge \(e\in M\). A minimum cut is a  cut whose total edge capacity is the minimum among all cuts.

\begin{definition}[MC Mechanism]  Given an acyclic digraph \(G = (V,E,c)\), let the reported capacity of each player $i \in E$ be $\hat c_i \in [0, c_i]$. After all players reporting their capacities, the MC mechanism goes in two steps.  (i) For all  $s$-$t$ edges  \(e_i\in E\), we allocate payoffs of their capacities, i.e., $MC_i(\hat c)=\hat{c}_i$. (ii) We remove all $s$-$t$ edges and use ${\mathbf M}$ to denote the set of all minimal cuts of the remaining graph. Suppose the max-flow value of the remaining graph is 
\(\hat F\), and for each minimal cut $M$ of the remaining graph, we use  \(\hat c(M)=\sum\limits_{e\in M}\hat c(e)\) to denote its total capacity. 
For $e_i\in E$, its payoff is
\[MC_i(\hat c)=\frac{\hat F}{|{\mathbf M}|}\sum\limits_{M\in {\mathbf M}: M\ni e_i} \frac{\hat c(e)}{\hat c(M)}\].
\end{definition}

The formula in the first step is intuitive. Since for each player whose edge is an $s$-$t$ one, her contribution to any coalition is always her capacity, it's reasonable to allocate this value to her.

The formula in the second step is also intuitive. It could be understood as first allocating the total payoff \(\hat F\) evenly among all minimal cuts, i.e., each minimal cut receives \(\frac{\hat F}{|{\mathbf M}|}\),  and then, for each minimal cut, allocating the received \(\frac{\hat F}{|{\mathbf M}|}\) to the corresponding players proportional to their capacities.



\subsection{Main properties}

\begin{theorem}The MC mechanism is DSIC, SIR, SP, MP, and CM for max-flow games.\label{Monotonicity of MC}
\end{theorem}
\begin{proof}

{\bf DSIC:} For any player $i\in N$, given the reported capacities of other players, we claim that her payoff in the MC mechanism $MC_i$ is increasing in her own reported capacity $\hat c_i$.  The claim is true for any player owning a source-to-sink edge, because the allocated payoff is her reported capacity. The claim is also true for players owning a non-source-to-sink edge, because when  her reported capacity increases, both the max-flow value $F(c)$ and her proportion in every minimal cut increase too. This monotonicity demonstrates that no player has an incentive to under-report, and hence the MC mechanism is DSIC, because we have assumed that no player can over-report.

{\bf SIR:} To prove IR, i.e., $MC_i(\hat c)\geq \hat v(e_i)$, we only need to consider players possessing source-to-sink edges. Since they are allocated their stand-alone payoffs $\hat v(e_i)$ due to the first step of the MC mechanism, IR is obvious. Since we have assumed that each player's edge is on at least one $s$-$t$ path, to further show SIR, we need to prove that that $MC_i>0$ for every player $i$. This is indeed correct because the above assumption implies that every edge belongs to at least one minimal cut (after the removing of source-to-sink edges).

{\bf SP \& MP:} For any two parallel edges \(e_1\) and \(e_2\), we consider what happens when  they merge into a single edge \(e_{1,2}\) with \(c_{1,2}=c_1 + c_2\) and \(\hat c_{1,2}=\hat c_1 +\hat c_2 \). Observe that either they are both source-to-sink edges, or neither of them is. In the former case, the allocated payoff by MC is obviously the sum of their previous separate payoffs. Thus we only need to consider the latter case. Note in this case that this merge does not change the set of minimal cuts and the total capacity of any minimal cut and the max-flow value. This ensures that the MC mechanism is MP. SP can be shown in the same way.

{\bf CM:} The CM property of the MC mechanism follows directly from Theorem \ref{Cross-effect between edges} that will be introduced soone and thus we omit its proof.   
\end{proof}

Before comparing the MC mechanism  with the core, we recall a characterization of the core for max-flow games with unit-capacity edges. For these games,  the core is the convex hull of the indicator vectors of all the minimum cuts\footnote{See Theorem 4 of \cite{1982kalaibalanced}. Note that the term ``minimal cut'' in their paper refers  actually to ``minimum cut'' in most literature, including this paper. See also \cite{Deng2009nucleolus}.}. Though this property does not extend to general max-flow games\footnote{See Example 1 of the working paper version of \cite{1982kalaibalanced}.}, it is helpful for us to better understand the differences between core-selection mechanisms and the MC mechanism.

Compared with the many nice properties of the MC mechanism, no core-selection mechanism is NEIC. A key difference between a minimal cut and a minimum cut to explain this is that while a minimal cut is pure structural, i.e., it only depends on the topology of the network, a minimum cut depends also on the capacities of the edges. These capacities are private information and may not be truthfully reported. Specifically, players may have incentives to under-report their capacities to be included in more minimum cuts. In contrast, for the MC mechanism, players do not have incentives to under-report, because under-reporting does not change the set of minimal cuts, but harms their allocated payoff in each minimal cut.

The first step (i) is critical to the MC mechanism, without which it cannot be IR (and consequently cannot be SIR). To see this,  consider the example in Figure \ref{firststep}. There are 3 players in the game, among whom player 3 possesses an $s$-$t$ edge. If we skipped the first step, then we would have $MC_3=5/6<1$, which contradicts the definition of IR.

\begin{figure}[h]
\centering
\begin{tikzpicture}[scale=3]
   \fill (1,0) circle (1pt) node[below] {$s$};
   \fill (2,0) circle (1pt) node[below] {$A$};
   \fill (3,0) circle (1pt) node[below] {$t$};
  \draw[line width=0.8pt, blue,postaction={decorate,decoration={markings,mark=at position 0.5 with {\arrow{>[scale=1]}}}}] (1,0) to[out=0, in=180] node[above] {$c_1=1$} (2,0);
  \draw[blue, line width=0.8pt, postaction={decorate,decoration={markings,mark=at position 0.5 with {\arrow{>[scale=1]}}}}] (2,0) to[out=0, in=180] node[above] {$ c_{2}=2$} (3,0);

  \draw[red, line width=2pt, postaction={decorate,decoration={markings,mark=at position 0.5 with {\arrow{>[scale=1]}}}}] (1,0) to[out=60, in=120] node[above] {\textcolor{black}{$c_3=1$}} (3,0);
\end{tikzpicture}
\caption{Necessity of the first step in the  MC mechanism.}
\label{firststep}
\end{figure}
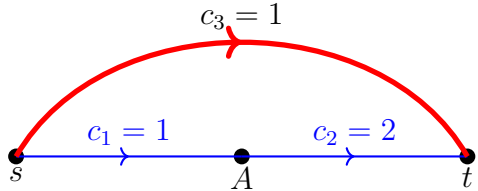

Combining  Theorem \ref{Monotonicity of MC} and Example \ref{counterexample}, we know that the allocation suggested by the MC mechanism may not be in the core for max-flow games, and thus players may have incentives to collude.

\subsection{More discussions on CM}
 We introduce two more terms to facilitate  our further analysis of the CM property.



\begin{definition}[Independent Edge Pairs]
A pair of edges \((e_1, e_2)\)  in $G$ is \emph{independent} if no minimal cut contains both of them. 
\end{definition}

\begin{definition}[Inclusive Edge Pairs]
 An ordered pair of edges \((e_1, e_2)\) in $G$ is  \emph{inclusive} if, for every minimal cut $M\ni e_2$, (i) $e_1\in M$, and (ii) $M\setminus \{e_1\}$ is a minimum cut of $G\setminus \{e_1\}$.
\end{definition}


While the concept of an independent edge pair is straightforward, the concept of an inclusive edge pair is more complex. To better understand the concept of inclusiveness, we introduce the terms 'critical value' \citep{shapley1961network} and 'essential edge' \citep{2023vvv}.  For  each edge $e$, its critical value is defined as \[x^*(e)= F_e(+\infty)-F_e(0).\] Intuitively, critical value of an edge is the capacity threshold such that the max-flow value is increasing below this threshold but keeps unchanged once it is exceeded. 
 An edge \(e\) is called  essential if the max-flow value of \(G\) strictly decreases when \(c(e)\) decreases, and inessential otherwise. 
Equivalently, \(e\) is an essential edge if and only if \(c(e)\leq x^*(e)\). The critical value, as well as essential and inessential edges,  is illustrated in  Figure \ref{flow function} via the max-flow function.

\begin{figure}[h]
\centering
\begin{tikzpicture}[scale=1.5]

  \draw[->] (-0.5,0) -- (3,0) node[right] {$c(e)$};
  \draw[->] (0,-0.5) -- (0,3) node[above] {$F_e(c(e))$};
 
  \draw[line width=0.8pt, domain=0:1.5,smooth,variable=\x,blue] (0,0.5)--(1.5,2)
    node[midway, below] {Essential};
\draw[line width=0.8pt, domain=1.5:3,smooth,variable=\x,red] (1.5,2)--(3,2) node[midway, below] {Inessential};
 
   \fill (0,0.5) circle (1 pt) node[left] {$F_{e}(0)$};
  \filldraw[blue] (1.5,2) circle (1pt) node[above, black] {};
  \filldraw[black] (1.5,0) circle (1pt) node[below] {$x^*(e)$};
  \filldraw[black] (0,2) circle (1pt) node[left] {$F_{e}(0)+x^*(e)$};
  \draw[dashed] (1.5,2) -- (1.5,0) ;
  \draw[dashed] (1.5,2) -- (0,2) ;
\end{tikzpicture}
\caption{The max-flow function \citep{shapley1961network}.}
    \label{flow function}
\end{figure}
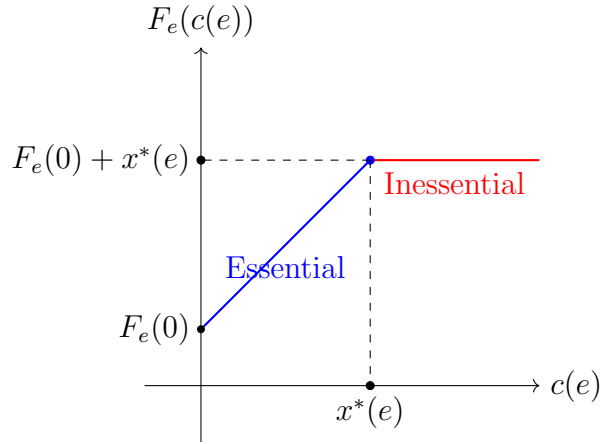

 When an ordered pair of edges \((e_1, e_2)\) of $G$ is inclusive, it means two things. First, the family of minimal cuts containing $e_2$ is a subfamily of minimal cuts containing $e_1$. Thus, when the capacity of $e_1$ changes, the total capacity value of each minimal cut containing $e_2$ is affected accordingly. However, this does not ensure that the max-flow value also changes. Second, if we delete $e_1$ in the graph, then each minimal cut containing $e_2$ in $G$ is actually a minimum cut in $G\setminus \{e_1\}$.

It can be verified that two edges are independent if they are connected in series (Figure \ref{C-I}), and inclusive if they are connected in parallel (Figure \ref{S-I}). A pair of edges may be neither independent nor inclusive. In fact,  this relationship is more typical than the discussed one: two edges are generally neither independent nor inclusive.

\begin{theorem}
Let the reported capacity of edge \(e_1\) increase, we analyze its impact on $MC_2$, the payoff allocated to  \(e_2\) in the MC mechanism. If either of them is an s-t edge,  $MC_2$ remains unchanged. Suppose that neither of them is an s-t edge in the following.
   \begin{enumerate}[(i)]
       
       \item If \((e_1, e_2)\) is independent, then \(MC_2(\hat c)\) is strictly increasing in \(\hat c_1\in (0,x^*(e_1)]\) and  remains constant in \(\hat c_1\in (x^*(e_1),\infty)\).

       \item If \((e_1, e_2)\) is inclusive, then \(MC_2(\hat c)\)  remains constant in \(\hat c_1\in (0,x^*(e_1)]\) and is strictly decreasing in \(\hat c_1\in (x^*(e_1),\infty)\).

       \item Otherwise, \(MC_2(\hat c)\) is strictly increasing  in \(\hat c_1\in (0,x^*(e_1)]\) and strictly decreasing in \(\hat c_1\in (x^*(e_1),\infty)\).
   \end{enumerate}
   \label{Cross-effect between edges}
\end{theorem}

\begin{proof}

(i) Suppose now that \((e_1, e_2)\) are independent. In this case, any minimal cut \(M\) containing \(e_2\) does not contain \(e_1\). The first part of the result follows from the fact that player 2's proportion in every minimal cut remains unchanged, while the max-flow value increases. The second part follows from the observation that the increase in \(\hat{c}_1\) does not affect either the max-flow value or player 2's proportion in any minimal cut.

(ii) Suppose now that \((e_1, e_2)\) is inclusive. Let $M$ be any minimal cut  containing $e_2$. By the definition of inclusiveness, 
$M$ also contains \(e_1\), and $M\setminus \{e_1\}$ is a minimum cut of $G\setminus \{e_1\}$.
In other words, when the reported capacity of player 1 is 0, i.e., $\hat c_1'' =0$, $M$ is a minimum cut, $\hat F''=\hat c''(M)$.

   Now, consider the first part.  Let \(\hat c_1\) increase to \(\hat c_1' = \hat c_1+ x\in  (0,x^*(e_1)]\). Thus $e_1$ remains an essential edge both before and after the increase, indicating that the max-flow value also increases by $x$.

   Since \(\hat c_1, \hat c_1' \in  (0,x^*(e_1)]\), we have $\hat F =\hat F''+ \hat c_1=\hat c''(M) + \hat c_1=\hat c(M)$, $\hat F' =\hat F''+ \hat c_1'=\hat c''(M) + \hat c_1'=\hat c'(M)$. Thus $M$ is a minimum cut  both before and after the increase in $e_1$'s reported capacity. Therefore,
    \begin{eqnarray*}MC_2(\hat c)&=&\sum\limits_{e_2\in M \in {\mathbf M}} \frac{\hat F}{|{\mathbf M}|} \frac{\hat c_2}{\hat c(M)}\\
    &=&\sum\limits_{e_2\in M \in {\mathbf M}} \frac{\hat c_2}{|{\mathbf M}|}\\
&=&\sum\limits_{e_2\in M \in {\mathbf M}} \frac{\hat c_2}{|{\mathbf M}|}\frac{\hat F+x}{\hat c(M)+x}\\
   &=&MC_2(\hat c').\end{eqnarray*}

    The second part is trivial, because when  \(\hat c_1\) increases in \((x^*(e_1),\infty)\), the max-flow value keeps unchanged, but the proportion of the capacity of $e_2$ decreases in every minimal cut.

    (iii) Suppose that \((e_1, e_2)\) is neither independent nor inclusive. Since \((e_1, e_2)\) is not independent, there exists at least one minimal cut containing both of them, and hence the second part can be proved in the similar way as in the proof of (ii). We will now focus on the first part.

     Since \((e_1, e_2)\) is not inclusive, either (a) there exists a minimal cut $M_0\ni e_2$  such that  \(e_1\notin M_0\), or (b) there exists a minimal cut $M_0\ni e_1,e_2$ such that $M_0\setminus \{e_1\}$ is not a minimum cut of $G\setminus \{e_1\}$. We now consider these two cases separately.  Let player 1's reported capacity increase from $\hat c_1$ to $\hat c_1+x\in (0,x^*(e_1)]$. As in the proof of (ii), we only need to consider the change in $\frac{\hat F}{\hat c(M)}$, since $\hat c_2$ and $|\mathbf{M}|$ both remain unchanged. 
    
    (a) Suppose now that there exists a minimal cut $M_0\ni e_2$  such that   \(e_1\notin M_0\).  Player 2's payoff consists of two components: one from minimal cuts containing both $e_1, e_2$, and the other from minimal cuts not containing $e_1$.  Player 2's payoff from every minimal cut $M$ containing both $e_1$ and $e_2$ is non-decreasing in $(0,x^*(e_1)]$, because the capacity of any cut $\hat c(M)$ is no less than the max-flow value $\hat F$, which implies $\frac{\hat F}{\hat c(M)} \leq \frac{\hat F+x}{\hat c(M)+x}$. Player 2's payoff from minimal cuts $M_0$ not containing $e_1$ is strictly increasing since the max-flow value increases while player 2's proportion remains unchanged. In summary, player 2's payoff is strictly increasing in $\hat c(e_1)\in (0,x^*(e_1)]$, and the first part of (iii) holds in this case. 
    
    (b)  Suppose now that there exists a minimal cut $M_0\ni e_1,e_2$ such that $M_0\setminus \{e_1\}$ is not a minimum cut of $G\setminus \{e_1\}$. Suppose also that (a) does not hold, then player 2's payoff comes entirely from minimal cuts containing both $e_1, e_2$. As in the proof of (a), player 2's payoff from every minimal cut $M$, rather than $M_0$, containing both $e_1$ and $e_2$,  is non-decreasing in $(0,x^*(e_1)]$. 
    
   It remains to consider $M_0$. Let $\hat c''(e_1)=0$, then the capacity of $M_0$ exceeds the max-flow value, i.e., $\hat c''(M_0) > \hat F''$. Thus, we have \[\frac{\hat F}{\hat c(M_0)} =\frac{\hat F'' + \hat c_1}{\hat c''(M_0)+\hat c_1} < \frac{\hat F''+ \hat c_1+x}{\hat c''(M_0)+ \hat c_1+x}= \frac{\hat F'}{\hat c'(M_0)}.\]In summary, player 2's payoff is strictly increasing in $\hat c(e_1)\in (0,x^*(e_1)]$, and the first part of (iii) is correct in this case. 
\end{proof}

Theorem \ref{Cross-effect between edges} provides a more accurate version of the CM part of Theorem \ref{Monotonicity of MC}. The CM part of Theorem \ref{Monotonicity of MC} states that   \(MC_2(\hat c)\) increases as  \(\hat c_1\) increases and the max-flow value increases too, in which case it must be that  \(\hat c_1\) does not exceed the critical value \(x^*(e_1)\), i.e.,   \(\hat c_1\in (0,x^*(e_1)]\), because otherwise the max-flow value would not change. Theorem \ref{Cross-effect between edges} further refines this by stating that \(MC_2(\hat c)\) strictly increases if \((e_1, e_2)\) is independent, \(MC_2(\hat c)\) remains constant if \((e_1, e_2)\) is inclusive.

The CM part of Theorem \ref{Monotonicity of MC} does not say anything about the case that the capacity of $e_1$ increases but the max-flow value does not change, i.e., the case \(\hat c_1\in (x^*(e_1),\infty)\).  Theorem \ref{Cross-effect between edges} states that, in this case, \(MC_2(\hat c)\) either strictly decreases or remains constant, depending on the relationship between the two edges.


\section{Conclusion}
In this paper, we considered  a mechanism design problem in the max-flow setting where each player possesses an edge and the edge capacity is private information. We first showed that no core allocation is capable of soliciting the true information from the players and then designed a minimal cut based mechanism that is not only dominant strategy incentive compatible, but also satisfies several other natural properties. More discussions on core-selection mechanisms are provided in Appendix \ref{app:core}, which, together with previous discussions, gives a comparison of the major mechanism as summarized in Table \ref{tab:your_label}. 

To the best of our knowledge, we are the first to consider this problem, and the design of the mechanism based on minimal cuts is also novel. While we have assumed that the network has a single source and a single sink,  our main results can be generalized to allow multiple sources and/or multiple sinks. 


An obvious problem for future research is whether the MC mechanism can be computed efficiently, and if so, design efficient algorithms to compute it. Note that enumerating all minimal cuts is obviously inefficient, because the number of minimal cuts is not polynomially bounded.  The scenario that each player may possess multiple edges is also worth investigating.

We have assumed that no cost is incurred when a player's edge is used, and thus a player whose edge is not used at all may still receive a payoff larger than someone whose edge is used. This is of course not reasonable when there are edge costs associated with their use. It is natural to assume that the costs are also private information, which makes the mechanism design problem closer to the minimum-cost flow problem.  This interesting setting will be explored in the future.

\bibliography{main}

\newpage

\appendix

\section{More Discussions on Core-selection Mechanisms}\label{app:core}

\begin{proposition}There exists a core-selection mechanism that satisfies both SP and MP. 
\end{proposition}

\begin{proof} Given a digraph and reported capacities, we choose a minimum cut  $M$ that is nearest to the source. Consider the following mechanism: The payoff allocated to any player is her capacity if her controlled edge is in $M$, and $0$ otherwise. 
It is known that this allocation is in the core \citep{1982kalaibalanced} and hence this is a core-selection mechanism. To show that this mechanism is SP and MP, it is sufficient to observe  the following two facts. (i)  Whenever an edge is  split into two parallel ones, this does not affect the minimum cut nearest to the source, and the old edge is in the old $M$ if and only if the two new edges are both in the new $M$, in which case the total allocated payoff remains unchanged. (ii) Whenever two parallel edges are merged into one, this does not affect the minimum cut nearest to the source either, and the new edge is in the new $M$ if and only if the two old edges are both in the old $M$, in which case the total allocated payoff remains unchanged too. 
\end{proof}

\vspace{24pt}

\begin{example}[Core violates CM]
Consider  a four-player game in Figure \ref{Core violates CM}. When the edge capacity vector is $(1,1,1,1)$, the core  is \(\{(a, a, 1-a, 1-a): a \in [0,1]\}\), and thus there exists one player who has a payoff smaller than 1 (for any core-selection mechanism). We assume w.l.o.g. that player 2 is such a player. 
 When player 1's capacity is 0, however, the unique core allocation is $(0,1,0,0)$, in which the profit allocated to player 2 is 1. Therefore, when player 1 increases her capacity from 0 to 1, the max-flow value increases from 1 to 2, but the payoff allocated to  player 2 decreases. This  illustrates that no core-selection mechanism  satisfies CM.\end{example}

\begin{figure}[H]
 \centering
\begin{tikzpicture}[scale=5]
    \fill (1,0) circle (0.5pt) node[below] {$s$};
    \fill (2,0) circle (0.5pt) node[below] {$A$};
    \fill (3,0) circle (0.5pt) node[below] {$t$};

  \draw[red, line width=2pt, postaction={decorate,decoration={markings,mark=at position 0.5 with {\arrow{>[scale=1]}}}}] (1,0) to[out=60, in=120] node[above] {\textcolor{black}{$c_1=0 \rightarrow 1$}} (2,0);
 \draw[blue, line width=0.8pt, postaction={decorate,decoration={markings,mark=at position 0.5 with {\arrow{>[scale=1]}}}}](1,0) to[out=-60, in=-120] node[below] {\textcolor{black}{$c_2=1$}} (2,0);
  \draw[blue, line width=0.8pt, postaction={decorate,decoration={markings,mark=at position 0.5 with {\arrow{>[scale=1]}}}}] (2,0) to[out=60, in=120] node[above] {\textcolor{black}{$c_3=1$}} (3,0);
   \draw[blue, line width=0.8pt, postaction={decorate,decoration={markings,mark=at position 0.5 with {\arrow{>[scale=1]}}}}] (2,0) to[out=-60, in=-120] node[below] {\textcolor{black}{$c_4=1$}} (3,0);
\end{tikzpicture}

\caption{Core violates CM}
\label{Core violates CM}
\end{figure}
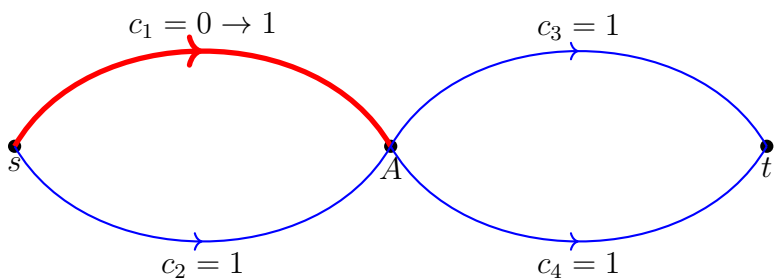

\end{document}